\documentclass[letterpaper]{article} 
\usepackage{aaai25}  
\usepackage{times}  
\usepackage{helvet}  
\usepackage{courier}  
\usepackage[hyphens]{url}  
\usepackage{graphicx} 
\usepackage{natbib}  
\usepackage{caption} 
\usepackage[linesnumbered,ruled,vlined]{algorithm2e}

\usepackage{todonotes}
\presetkeys{todonotes}{inline}{}
\usepackage{times}
\usepackage{soul}
\usepackage{url}
\usepackage[utf8]{inputenc}
\usepackage{graphicx}
\usepackage{xcolor}
\usepackage{amsmath}
\usepackage{booktabs}
\usepackage{enumitem}
\usepackage{dsfont}
\usepackage{enumerate}
\usepackage{nicefrac}
\usepackage{amsthm}
\usepackage{cleveref}
\usepackage{amssymb}
\usepackage{algpseudocode}
\usepackage{pifont}

\newtheorem{theorem}{Theorem}[section]
\newtheorem{corollary}[theorem]{Corollary}
\newtheorem{proposition}[theorem]{Proposition}

\newtheorem{observation}[theorem]{Observation}

\theoremstyle{definition}
\newtheorem{definition}[theorem]{Definition}
\newtheorem{example}[theorem]{Example}
\newtheorem{remark}[theorem]{Remark}

\title{Verifying Proportionality in Temporal Voting}
\author {
    Edith Elkind\textsuperscript{\rm 1},
    Svetlana Obraztsova\textsuperscript{\rm 2},
    Jannik Peters\textsuperscript{\rm 3},
    Nicholas Teh\textsuperscript{\rm 4}
}
\affiliations {
    \textsuperscript{\rm 1}Northwestern University, USA\\
    \textsuperscript{\rm 2}Carleton University, Canada\\
    \textsuperscript{\rm 3}National University of Singapore, Singapore\\
    \textsuperscript{\rm 4}University of Oxford, UK\\
    edith.elkind@northwestern.edu, svetlanaobraztsova@cunet.carleton.ca, peters@nus.edu.sg, nicholas.teh@cs.ox.ac.uk
}

\usepackage{bibentry}

\begin{document}

\maketitle

\begin{abstract}
    We study a model of temporal voting where there is a fixed time horizon, and at each round the voters report their preferences over the available candidates and a single candidate is selected. Prior work has adapted popular notions of justified representation as well as voting rules that provide strong representation guarantees from the multiwinner election setting to this model. In our work, we focus on the complexity of verifying whether a given outcome offers proportional representation. We show that in the temporal setting verification is strictly harder than in multiwinner voting, but identify natural special cases that enable efficient algorithms.
\end{abstract}

\section{Introduction}
Consider a large corporation that has decided to improve its public image and to give back to the society by engaging in \emph{corporate philanthropy} over the next decade.
They will commit a small fraction of their profits towards supporting the efforts of a single charitable organization, to be selected on an annual basis.
The management of this corporation decides to ask its customers, staff, and shareholders for input as to which charity organization it should select each year.
Furthermore, the charity selected is of strategic importance---it would directly impact the company's corporate image and hence profitability.
As such, it is important for the company to ensure that the selection is representative of what its customers, staff, and shareholders care about and that it would create maximum impact for the charity organization they choose to support.

It is natural to view this problem through the lens of  
multiwinner voting \cite{faliszewski2017mwv,elkind2017propertiesmwv,lackner2022abc}:
indeed, the corporation's goal is to select a fixed-size subset of {\em candidates}
(in this case, charities) while respecting the preferences of \emph{agents} (in this case, the customers, staff, and shareholders). In this context, several notions of representation and fairness have been proposed over the past decade, spanning across proportional representation~\cite{elkind2017propertiesmwv,aziz2020expanding}, diversity \cite{bredereck2018diversity,celis2018fairness}, and excellence, amongst others \cite{lackner2022abc}.
Perhaps the most prominent among these is the concept of \emph{justified representation (JR)} and its variants (such as proportional justified representation (PJR) and extended justified representation (EJR)), which aim to capture the idea that large cohesive groups of voters should be fairly represented in the final outcome \cite{aziz2017jr,sf2017pjr,aziz2018complexityepjr,peters2021fjr,Brill023}.

However, the existing notions of fairness do not fully capture the complexity of our setting: traditional multiwinner voting models consider decisions made in a single round, with the entire set of candidates to fund  (or candidates) being chosen simultaneously.
In contrast, in our model the decisions are made over time, voters' preferences 
may evolve, and a candidate may be chosen multiple times. This calls for adapting
the JR axioms to the temporal setting.

Temporal considerations in the multiwinner voting setting have been studied recently, most notably in the line of work known as \emph{perpetual voting} \cite{lackner2020perpetual,lackner2023proportionalPV} or \emph{temporal voting} \cite{elkind2024temporalelections,zech2024multiwinnerchange}, and have broad real-world applications (see the recent survey by Elkind et al.~\shortcite{elkind2024temporal}). 
In particular, Bulteau et al.~\shortcite{bulteau2021jrperpetual} and, subsequently, 
Chandak et al.~\shortcite{chandak23} defined temporal analogues of the justified
representation axioms and investigated whether existing multiwinner rules with 
strong axiomatic properties can be adapted to the temporal setting so as to
satisfy the new axioms.

\subsection{Our Contributions}
We build upon the works of Bulteau et al.~\shortcite{bulteau2021jrperpetual}
and Chandak et al.~\shortcite{chandak23}, and focus on the complexity 
of {\em verifying} whether a given solution satisfies 
justified representation axioms. This task is important if, e.g., 
the outcome is fully or partially determined by external considerations, 
so explicitly using an algorithm to obtain an outcome with strong representation guarantees is not feasible, 
but representation remains an important concern.
In multiwinner voting setting, the verification problem is known to be coNP-hard 
for PJR and EJR, but polynomial-time solvable for JR. 
We argue that the existing complexity results do not automatically
transfer from the multiwinner setting to the temporal setting. However, 
we develop new proofs specifically tailored to the temporal setting, 
and show that in temporal elections all three
properties are coNP-hard to verify, even under strong constraints on the structure
of the input instance. Our complexity result for JR 
shows that the temporal setting is strictly harder than the multiwinner setting.

We complement our hardness results by fixed parameter tractability results
as well as a polynomial-time algorithm for a natural special case of our model
where candidates may join the election over time, but never leave, and voters'
preferences over available candidates do not change.
We also develop an integer linear programming formulation for the problem
of selecting an outcome that provides EJR
and satisfies additional linear constraints on voters' utilities, 
thereby establishing that this problem is fixed-parameter tractable with respect 
to the number of voters $n$.

Finally, we (partially) answer an open question of Chandak et al.~\shortcite{chandak23}, 
by showing that the prominent Greedy Cohesive Rule~\cite{bredereckF0N19,peters2021fjr}  can be adapted to the temporal setting.

\subsection{Related Work}
Our work belongs to the stream of research on \emph{perpetual}, or {\em temporal voting} \cite{lackner2020perpetual,lackner2023proportionalPV,elkind2024temporalelections}.
This line of work started by considering temporal extensions of popular multiwinner voting rules and simple axiomatic properties.
Bulteau et al.~\shortcite{bulteau2021jrperpetual} were the first to adapt notions of proportional representation to the temporal setting. They 
consider several temporal variants of JR and PJR, and 
show that a JR outcome can be computed in polynomial time even for the most
demanding notion of JR among the ones they consider; for PJR, 
they prove the existence of an outcome satisfying the axiom, but their
proof, while constructive, is based on an exponential-time algorithm. Page et al.~\shortcite{page2021electing} proposed very similar concepts in the context of electing the executive branch. This work is extended by Chandak et al.~\shortcite{chandak23}, who also propose a temporal variant of EJR, 
and show that several multiwinner voting rules
that satisfy PJR/EJR in the standard model can be extended
to the temporal setting so as to satisfy temporal variants of PJR/EJR.
Bredereck et al.~\shortcite{bredereck2020successivecommittee,bredereck2022committeechange} look at sequential committee elections whereby 
an entire committee is elected in each round, and impose
constraints on the extent a committee can change, whilst ensuring that the candidates retain sufficient support from the electorate.
Zech et al.~\shortcite{zech2024multiwinnerchange} study a similar model, but focus on the class of Thiele rules.
Elkind et al.~\shortcite{elkind2024temporal} offer a systematic review of recent work 
on temporal voting.

The temporal voting model captures apportionment with approval preferences \cite{brill2024partyapportionment,delemazure2022spelection}, by endowing the voters with preferences that do not change over time.
Other models in the social choice literature that include temporal elements include
public decision-making \cite{alouf2022better,conitzer2017fairpublic,fain2018publicgoods,skowron2022proppublic,lackner2023freeriding,MPS23a,neoh2025strategicchores}, scheduling \cite{elkind2022temporalslot,patro2022virtualconf}, resource allocation over time \cite{bampis2018fairtime,allouah2023fairallocationtime,elkind2024temporalfairdivision}, online committee selection \cite{do2022onlinecommittee}, and dynamic social choice \cite{parkes2013dynamicsocialchoice,freeman2017dynamicsocialchoice}.

Lackner et al.~\shortcite{lackner2021longtermpb} propose a framework for studying long-term participatory budgeting, and study fairness considerations in that setting.

\section{Preliminaries}\label{sec:prelim}
We start by introducing the basic model of multiwinner voting with approval 
ballots.
A {\em multiwinner approval election} is given 
by a set of {\em candidates} $P$, a set of {\em voters} $N$, a list of {\em approval sets} $(s_i)_{i\in N}$, 
where $s_i\subseteq P$ for each $i\in N$, and an integer $k$; 
the goal is to select a {\em committee}, i.e., a size-$k$ subset of $P$.
There is a large body of research on axioms
and voting rules for this model~\cite{lackner2022abc}.

In this work, we consider temporal voting with approval ballots, as 
defined by Bulteau et al.~\shortcite{bulteau2021jrperpetual} and Chandak et al.~\shortcite{chandak23}.
A {\em temporal election} is a tuple $(P, N, \ell, ({\mathbf s}_i)_{i\in N})$,
where $N = \{1,\dots,n\}$ is a set of \emph{voters}, 
      $P = \{p_1,\dots,p_m\}$ is a set of $m$ {\em candidates}, 
       $\ell$ is the number of \emph{rounds}, and, 
       for each $i\in N$, ${\mathbf s}_i = (s_{i,1},s_{i,2},\dots, s_{i,\ell})$, 
       where $s_{i, t}\subseteq P$ is the {\em approval set} of voter $i$
       in round $t$, which consists of candidates that $i$ approves in round $t$.
We refer to $\mathbf{s}_i$ as $i$'s \emph{temporal preference};
for brevity, we will sometimes omit the term ``temporal''.
An \emph{outcome} of a temporal election $(P, N, \ell, ({\mathbf s}_i)_{i\in N})$ is a sequence $\mathbf{o} = (o_1,\dots,o_\ell)$
of $\ell$ candidates such that for every $t \in [\ell]$ candidate $o_t \in P$ is chosen in round $t$. 
A candidate may be selected multiple times, i.e., it may be the case that $o_t=o_{t'}$ for $t\neq t'$.
A voter $i$'s \textit{satisfaction} from an outcome $\mathbf{o}$ is computed as $\textit{sat}_i(\mathbf{o}) = |\{t \in [\ell]: o_t \in s_{i,t}\}|$.

An important concern in the multiwinner setting is group fairness, i.e., making
sure that large groups of voters with similar preferences are represented by
the selected committee. The most well-studied group fairness axioms are (in increasing
order of strength) {\em justified representation (JR)}~\cite{aziz2017jr}, {\em proportional justified representation (PJR)}~\cite{sf2017pjr}, and {\em extended justified representation (EJR)}~\cite{aziz2017jr}.
We will now formulate these axioms, as well as their extensions
to the temporal setting; for the latter, we follow the terminology of Chandak et al.~\shortcite{chandak23}. Definition~\ref{def:jr-static} may appear
syntactically different from the standard definitions of these notions, but 
it can easily be shown to be equivalent to them; further, it has the advantage
of being easily extensible to the temporal setting.

\begin{definition}\label{def:jr-static}
For a multiwinner election $(P, N, (s_i)_{i\in N}, k)$
and a group of voters $N'\subseteq N$, we define
the {\em demand} of $N'$ as 
$$
\alpha^{\textit{mw}}(N') = \min\{\beta^{\textit{mw}}(N'), \gamma^{\textit{mw}}(N')\}, \quad \text{where}
$$ 
$$
\beta^{\textit{mw}}(N')= |\cap_{i\in N'}s_i| \quad\text{ and }\quad
\gamma^{\textit{mw}}(N')=\left\lfloor  k\cdot \frac{|N'|}{n}\right\rfloor.
$$ 
A committee $W\subseteq P$ provides {\em justified representation (JR)} 
if for every 
$N'\subseteq N$ with $\alpha^{\textit{mw}}(N')>0$ 
there is an $i\in N'$ such that
$|s_i\cap W|>0$;
it provides {\em proportional justified representation (PJR)} 
if for every $N'\subseteq N$ we have
$|(\cup_{i\in N'} s_i)\cap W|\ge \alpha^{\textit{mw}}(N')$, and
it provides {\em extended justified representation (EJR)} 
if for every $N'\subseteq N$ there is an $i\in N'$ such that
$|s_i\cap W|\ge \alpha^{\textit{mw}}(N')$.
\end{definition}

When extending these notions to the temporal setting, 
Bulteau et al.~\shortcite{bulteau2021jrperpetual} consider JR and PJR (but not EJR), 
each with three variants---prefixed with ``static'', 
``dynamic all-periods-intersection'', and ``dynamic some-periods-intersection''; where 
``dynamic some-periods intersection'' is the most demanding of these.
Chandak et al.~\shortcite{chandak23} extend this analysis to EJR, and focus on two variants
of the axioms: ``dynamic all-period intersection'', which they call weak
JR/PJR/EJR, and ``dynamic some-period intersection'', which they call  JR/PJR/EJR\footnote{The conference version of their paper uses ``JR/PJR/EJR'' for the weaker version and ``strong JR/PJR/EJR'' for the stronger version; we use the terminology from the arXiv version of their paper.}. In what follows, we use the terminology 
of Chandak et al.~\shortcite{chandak23}. 

Just as in the multiwinner setting, 
for each group of voters $N'$ we determine its demand $\alpha(N')$, 
which depends both on the size of $N'$ and on the degree of agreement among the group
members. The axioms then require that the collective satisfaction of group members
is commensurate with the group's demand.

\begin{definition}
    Given an election $E = (P, N, \ell, ({\mathbf s}_i)_{i\in N})$ and a group
    of voters $N'\subseteq N$, we define  the {\em agreement} of $N'$ 
    as the number of rounds in which all members of $N'$ approve a common candidate:
    $$
    \beta(N') = |\{t\in [\ell]: \bigcap_{i\in N'} s_{i, t}\neq\varnothing\}|.  
    $$
    We define the {\em demand} of a group of voters $N'$ as 
    $$
    \alpha(N') = \left\lfloor\beta(N')\cdot\frac{|N'|}{n}\right\rfloor.
    $$
\end{definition}
That is, if voters in $N'$ agree in $\beta$ rounds, 
they can demand a fraction of these rounds that is proportional to~$|N'|$.

We now proceed to define temporal extensions of JR, PJR, and EJR, starting with their stronger versions

\begin{definition}[Justified Representation]\label{def:jr}
   An outcome $\mathbf{o}$ provides \emph{justified representation (JR)} for
   a temporal election $E = (P, N, \ell, ({\mathbf s}_i)_{i\in N})$ if for every
   group of voters $N'\subseteq N$ with $\alpha(N')> 0$ 
   we have $\textit{sat}_i({\mathbf o})>0$ for some $i\in N'$.
\end{definition}

\begin{definition}[Proportional Justified Representation]\label{def:pjr}
   An outcome $\mathbf{o}$ provides \emph{proportional justified representation (PJR)}  for 
   a temporal election $E = (P, N, \ell, ({\mathbf s}_i)_{i\in N})$ if    for every group of voters $N'\subseteq N$  
   it holds that $|\{t\in [\ell]: o_t\in \bigcup_{i\in N'} s_{i,t}\}| \ge \alpha(N')$.
\end{definition}

\begin{definition}[Extended Justified Representation]\label{def:ejr}
   An outcome $\mathbf{o}$ provides \emph{extended justified representation (s-EJR)} for 
   a temporal election $E = (P, N, \ell, ({\mathbf s}_i)_{i\in N})$ if
   for every group of voters $N'\subseteq N$ there exists a voter $i\in N'$ with
   $\textit{sat}_i({\mathbf o})\geq \alpha(N')$.
\end{definition}

\noindent Note that EJR implies PJR, and PJR implies JR.

It is instructive to compare Definitions~\ref{def:jr}--\ref{def:ejr} to Definition~\ref{def:jr-static}. In particular, one may wonder why we did not define 
the demand of a group $N'$ in the temporal setting
as $\overline{\alpha}(N') = \min\{\beta(N'), \ell\cdot\frac{|N'|}{n}\}$, 
thereby decoupling the constraints on the size of the group and the level of agreement.
Note that $\beta(N')\le \ell$ for all $N'\subseteq N$ and hence 
$\overline{\alpha}(N')\ge \alpha(N')$ for all $N'$.
However, the following example shows that 
if we were to use this definition of demand, 
even the JR axiom would be impossible to satisfy
(whereas Chandak et al.~\shortcite{chandak23} show that every temporal election
admits an outcome that provides EJR); see also the discussion in Section~5
of the paper by Chandak et al.~\shortcite{chandak23}.

\begin{example}
Let $N=\{1, \dots, 6\}$, and let $\mathcal T$
be the set of all size-2 subsets of $N$ (so $|{\mathcal T}|=15$). 
Consider an election with voter set $N$, 
$P=\{x_T\}_{T\in{\mathcal T}}\cup\{y_j\}_{j=1, \dots, 6}$, and
$\ell=3$, where the voters' approval sets have the following structure:
$$
s_{i, 1}=\{x_T: i\in T\}, \quad s_{i, 2}=s_{i, 3} = \{y_i\} \quad\text{ for all $i\in N$}.
$$
For each group of voters $N'$ with $|N'|=2$ we have $\beta(N')=1$, 
as both voters in $N'$ approve $x_{N'}$ in the first round.
Moreover, $\ell\cdot\frac{|N'|}{n}=1$, so $\overline{\alpha}(N')=1$.
Hence, if we were to replace $\alpha(N')$ with $\overline{\alpha}(N')$ in the definition of JR, we would have to ensure that for every pair of voters $N'$ at least one voter in $N'$ obtains positive satisfaction. However, there is no outcome $\mathbf o$ that accomplishes this: $o_1$ is approved by at most two voters, and $o_2$ and $o_3$ are approved by at most one voter each, so for every outcome $\mathbf o$ there are two voters whose satisfaction is 0. That is, the modified definition
of JR (and hence PJR and EJR) is unsatisfiable.

On the other hand, we have $\beta(N')=3$ if $|N'|=1$, $\beta(N')=1$ if $|N'|=2$
and $\beta(N')=0$ if $|N'|\ge 3$, and hence $\alpha(N')=0$ for all $N'\subseteq N$.
Thus, every outcome $\mathbf o$ provides EJR (and hence PJR and JR).
\end{example}

The axioms introduced so far apply to groups of voters with positive agreement.
Bulteau et al.~\shortcite{bulteau2021jrperpetual} and Chandak et al.~\shortcite{chandak23} also consider weaker axioms, which only apply to groups of voters that agree in all rounds. 

\begin{definition}[Weak Justified Representation/Proportional Justified Representation/Extended Justified Representation]\label{def:weak-jr}
   Consider an outcome $\mathbf{o}$ for a temporal election $E = (P, N, \ell, ({\mathbf s}_i)_{i\in N})$. We say that $\mathbf o$ provides:
   \begin{itemize}
   \item
   \emph{weak justified representation} (w-JR) if 
   for every group of voters $N'\subseteq N$ with $\beta(N')=\ell$, $\alpha(N')>0$  
   it holds that $\textit{sat}_i({\mathbf o})>0$ for some $i\in N'$.
   \item 
   \emph{weak proportional justified representation} (w-PJR) if 
   for every group of voters $N'\subseteq N$ with $\beta(N')=\ell$ 
   it holds that $|\{t\in [\ell]: o_t\in \bigcup_{i\in N'} s_{i,t}\}| \ge \alpha(N')$.
   \item
   \emph{weak extended justified representation} (w-EJR) if 
   for every group of voters $N'\subseteq N$ with $\beta(N')=\ell$ 
   there exists a voter $i\in N'$ with
   $\textit{sat}_i({\mathbf o}) \geq \alpha(N')$.
   \end{itemize}
\end{definition}
By construction, the weak JR/PJR/EJR axioms are less demanding than their regular counterparts: e.g., if $\ell=10$ and all voters agree on a candidate in each of the first 9 
rounds, but each voter approves a different candidate in the last round, these
axioms are not binding. Indeed, they are somewhat easier to satisfy: e.g., 
Chandak et al.~\shortcite{chandak23} show that a natural adaptation of the 
Method of Equal Shares~\cite{peters2020proportionality} satisfies w-EJR, but not EJR. Therefore, in our work we consider both families of axioms.

\section{Hardness Proofs}\label{sec:hard}
In multiwinner elections, outcomes that provide
EJR (and thus JR and PJR) can be computed in polynomial time~\cite{aziz2018complexityepjr,peters2020proportionality}.
In contrast, the problem of verifying if a given outcome 
provides JR/PJR/EJR is considerably more challenging: while
this problem is polynomial-time solvable for JR, it is coNP-hard 
for PJR and EJR~\cite{aziz2017jr,aziz2018complexityepjr}.

In this section, we show that in the temporal setting the verification
problem is coNP-hard even for (w-)JR. These results extend to (w-)PJR and (w-)EJR.

We note that the hardness results for PJR and EJR in the multiwinner setting 
do not transfer immediately to the temporal setting. This is because the notion of agreement
in our model is fundamentally different from the one in the multiwinner model:
we focus on the number of rounds in which a group of voters agrees on a candidate, 
whereas in the multiwinner model the cohesiveness of a group is determined
by the number of candidates the group members agree on. Suppose that we naively transform 
a multiwinner instance $(P, N, (s_i)_{i\in N}, k)$, where $k$ is the number of candidates
to be elected, into a temporal instance with $k$ rounds 
where voters' approvals in each round are given by $(s_i)_{i\in N}$. Consider a group
of voters $N'$ that only agrees on a single candidate in the original instance (and therefore
can demand at most one spot in the elected committee). 
In the temporal setting, voters in $N'$
will agree on that candidate in all rounds, and therefore if $N'$ is large,  
its demand can be substantial.
Thus, we have to prove hardness of verifying PJR and EJR from scratch.
Fortunately, the proofs of Theorems~\ref{thm:np-hard-ap} and~\ref{thm:np-hard} 
apply to all three representation concepts that we consider.

\begin{theorem}\label{thm:np-hard-ap}
    For each of X $\in\{$w-JR, w-PJR, w-EJR$\}$, verifying whether an outcome provides X is {\em coNP}-complete.
    The hardness result holds for w-JR and w-PJR even if $|P|=3$, and for w-EJR even if $|P|=2$.
\end{theorem}
\begin{proof}
If an outcome $\mathbf o$ 
does not provide w-JR, this can be witnessed by a
group of voters $N'\subseteq N$ with $\beta(N')=\ell$, $\alpha(N')> 0$  
such that 
$\mathit{sat}_i({\mathbf o})=0$ for all $i\in N'$.
Hence, our problem is in coNP.
Similarly, for w-PJR (respectively, w-EJR) it suffices to exhibit a group of voters
$N'$ such that $\beta(N')=\ell$
and $|\{t\in [\ell]: o_t\in \cup_{i\in N'}s_{i, t}\}|< \alpha(N')$ (respectively,
$\beta(N')=\ell$, $\textit{sat}_i({\mathbf o})< \alpha(N')$ for all $i\in N'$).

To prove coNP-hardness for w-JR, we reduce from the NP-hard problem \textsc{Clique}.
An instance of \textsc{Clique} is given by a graph $G = (V,E)$ 
with vertex set $V$ and edge set $E$, together with a parameter $\kappa$;
it is a `yes'-instance if there exists a subset of vertices $V'\subseteq V$
of size $\kappa$ that forms a clique, i.e., 
for all $v, v'\in V'$ we have $\{v, v'\}\in E$.

Given an instance $(G, \kappa)$ of {\sc Clique} with $G=(V, E)$, 
where $V=\{v_1, \dots, v_\nu\}$,  
we construct an election with a set of voters $N = \{1, \dots, \nu\}\cup N_0$, 
where $|N_0|=(\kappa-1)\nu$, 
a set of candidates 
$P=\{p_1, p_2, p_3\}$, and $\ell = \nu$ rounds. 
For each $i \in [\nu]$, the preferences of voter $i$ are given by:
	\begin{equation*}
        s_{i,t}=
        \begin{cases} 
          \{p_2\} & \text{if } t=i; \\
          \{p_1,p_2\} & \text{if } \{v_i,v_t\} \in E; \\
          \{p_1\} & \text{otherwise}.
        \end{cases}
        \end{equation*}
        For each $i\in N_0$ and each $t\in [\ell]$ 
        we set $s_{i, t}=\{p_3\}$. 

    We claim that a set of voters $N'\subseteq [\nu]$ satisfies $\beta(N')=\ell$ if and only if
    the set $V'=\{v_i: i\in N'\}$ forms a clique in $G$. Indeed, let
    $V'$ be a clique in $G$, and consider a round $t\in [\ell]$. 
    If $t\in N'$, we have 
    $s_{t, t}=\{p_2\}$, $s_{i, t}=\{p_1, p_2\}$ for all $i\in N'\setminus\{t\}$, 
    i.e., $p_2\in \cap_{i\in N'}s_{i, t}$.
    On the other hand, if $t\not\in N'$, for each $i\in N'$ we have $s_{i, t}=\{p_1, p_2\}$
    if $\{v_i, v_t\}\in E$ and $s_{i, t}=\{p_1\}$ otherwise, 
    so $p_1\in \cap_{i\in N'}s_{i, t}$. Thus, for each $t\in [\ell]$ the set
    $\cap_{i\in N'}s_{i, t}$ is non-empty. Conversely, if
    $\{v_i, v_j\}\not\in E$ for some $i, j\in N'$, $i \neq j$, 
    then $s_{i, j}=\{p_1\}$, whereas $s_{j, j}=\{p_2\}$, 
    i.e., $\cap_{i'\in N'}s_{i', j}=\varnothing$ and hence $\beta(N')<\ell$.
   
    We claim that an outcome $\mathbf{o} = (p_3,\dots, p_3)$ 
 fails to provide w-JR
    if and only if $G$ contains a clique of size $\kappa$. 

    Indeed, let $V'$ be a size-$\kappa$ clique in $G$, and let $N'=\{i: v_i\in V'\}$.
    We have argued that $\beta(N')=\ell$ and hence 
    $\alpha(V')=\ell\cdot\frac{\kappa}{\kappa\cdot\nu}=1>0$; however,
    $\textit{sat}_i({\mathbf o})=0$ for each $i\in N'$.
    For the converse direction, suppose
    there exists a subset of voters $N'$ witnessing that $\mathbf o$ 
    fails to provide w-JR.
    As $\textit{sat}_i({\mathbf o})=0$ for each $i\in N'$, 
   we have  $N'\cap N_0=\varnothing$, i.e., $N'\subseteq [\nu]$. 
    Moreover, $\beta(N')= \ell$ and hence
    the set $V'=\{v_i: i\in N'\}$ forms a clique in $G$. 
    Then $\alpha(N')>0$ can be re-written as 
    $\ell\cdot\frac{|N'|}{\kappa\cdot \nu}\ge 1$, i.e., $|N'|\ge \kappa$. 
    As $|V'|=|N'|$, the set $V'$ is a clique of size at least $\kappa$ in~$G$.

    The argument extends easily to w-PJR/w-EJR: 
    if $G$ contains a clique of size $\kappa$, then, as argued above, 
    $\mathbf o$ fails to provide w-JR (and hence it fails to provide w-PJR/w-EJR), 
    whereas if there are no cliques of size $\kappa$ in $G$, there is no 
    group of voters $N'\subseteq N$ with $\beta(N')=\ell$, $\alpha(N')>0$, 
    so w-PJR/w-EJR is trivially satisfied.
     In the appendix, we provide a different proof for w-EJR, which applies even if $|P|=2$.
\end{proof}

For stronger versions of our axioms, we obtain a hardness result even for $|P|=2$.

\begin{theorem}\label{thm:np-hard}
	For each of X $\in\{$JR, PJR, EJR$\}$, 
	verifying whether an outcome provides X is {\em coNP}-complete.
        The hardness result holds even if $|P|=2$.
\end{theorem}

\section{Tractability Results for $\boldsymbol{|P|=2}$}\label{sec:m}
The hardness results in Theorems~\ref{thm:np-hard-ap} and~\ref{thm:np-hard}
hold even if the number of candidates $|P|$ is small (3 for w-JR/w-PJR and 2 for w-EJR/JR/PJR/EJR).
Now, for $|P|=1$ verification is clearly easy: there is only one possible
outcome, and for any group of voters $N'$ this outcome provides satisfaction
of at least $\beta(N')\ge \alpha(N')$ to all voters in $N'$.
To complete the complexity classification with respect to $|P|$,
it remains to consider the complexity of verifying w-JR/w-PJR when $|P|=2$.

\begin{proposition}\label{prop:P=2}
Given an election $(P, N, \ell, ({\mathbf s}_i)_{i\in N})$ with $|P|=2$
and an outcome $\mathbf o$, we can check in polynomial time whether 
$\mathbf o$ provides w-JR.
\end{proposition}
\begin{proof}
Assume that $P=\{p, q\}$.
 Suppose that a group of voters $N'$ witnesses that $\mathbf o$ fails
  to provide w-JR. Then $\mathit{sat}_i({\mathbf o})=0$ for each $i\in N'$, but $\beta(N')=\ell$. 
  
We say that a voter $i$ is {\em grumpy}
if in each round $t$ she approves a single candidate, and this candidate is not $o_t$, i.e., $|s_{i, t}|=1$ and $o_t\not\in s_{i, t}$ for all $t\in [\ell]$. Let $G$ be the set of all grumpy voters, and note that $\beta(G)=\ell$: in each round $t$, all grumpy voters approve the (unique) candidate in 
$P\setminus\{o_t\}$.

Note that $N'\subseteq G$: indeed, 
if a voter $i$ is not grumpy, either she approves $o_t$ for some $t\in [\ell]$, so $\mathit{sat}_i({\mathbf o})>0$, 
or she has $s_{i, t}=\varnothing$
for some $t\in [\ell]$, which is not compatible with $\beta(N')=\ell$.
    Hence, to check whether $\mathbf o$ provides w-JR 
    it suffices to verify whether 
    $\alpha(G)=\lfloor\ell\cdot\frac{|G|}{|N|}\rfloor = 0$.
\end{proof}
It is not clear if Proposition~\ref{prop:P=2} can be extended to w-PJR:
to verify w-PJR, it is no longer sufficient to focus
on grumpy voters. We conjecture that verifying w-PJR remains
hard even for $|P|=2$
(recall that, by Theorem~\ref{thm:np-hard-ap}, this is the case for w-EJR).

Observe that for the election constructed in the proof of Theorem~\ref{thm:np-hard-ap} 
we have $s_{i, t}\neq\varnothing$ for all $i\in N$, $t\in [\ell]$. In contrast, 
in the election constructed in the proof of Theorem~\ref{thm:np-hard}  
the approval sets $s_{i, t}$ may be empty. It turns out that 
if we require that $s_{i, t}\neq\varnothing$ for all $i\in N$, $t\in [\ell]$ then for $|P|=2$ the problem of verifying JR becomes easy; under this assumption, if $\textit{sat}_i({\mathbf o})=0$ for all $i\in N'$ then all voters in $N'$ are grumpy, so it suffices to check if $\alpha(G)=0$.

\begin{proposition}\label{prop:nonempty}
Given an election $(P, N, \ell, ({\mathbf s}_i)_{i\in N})$ with $|P|=2$
and $s_{i, t}\neq\varnothing$ for all $i\in N$, $t\in [\ell]$ 
and an outcome $\mathbf o$, 
we can check in polynomial time whether $\mathbf o$ provides JR.  
\end{proposition}

Again, it is not clear if Proposition~\ref{prop:nonempty} extends to (w-)PJR and (w-)EJR; we conjecture that the answer is `no'.

\begin{remark}
    The proofs of
    Propositions~\ref{prop:P=2} and~\ref{prop:nonempty} go through if, instead of requiring $|P|=2$, we require that in each round there are at most two candidates that receive approvals from the voters, i.e.,   $|\cup_i s_{i, t}|\le 2$ for each $t\in [\ell]$. 
\end{remark}
    
Proposition~\ref{prop:nonempty} shows that, for $|P|=2$, requiring each voter to approve at least one 
candidate in each round reduces the complexity of checking JR considerably. However, this requirement only helps if $|P|$ is bounded.
Indeed, we can modify the construction 
in the proof of Theorem~\ref{thm:np-hard} by creating a set of additional candidates $P^N = \{p^i\}_{i\in N}$ and modifying the preferences so that whenever $s_{i, t}=\varnothing$ in the original construction, we set $s_{i, t}=\{p^i\}$. The new candidates 
do not change the agreement of any voter group, so the rest of the proof goes through unchanged.

\begin{proposition}\label{prop:verify_coNPhard_nonempty}
	For each of X $\in\{$JR, PJR, EJR$\}$,
	verifying whether an outcome provides X is {\em coNP}-complete,
        even if $s_{i, t}\neq\varnothing$
	for all voters $i\in N$ and all rounds $t\in [\ell]$.
\end{proposition}

It remains an open question whether 
verifying that a given outcome provides JR/PJR/EJR remains coNP-complete if all approval sets are non-empty and 
the size of the set $P$ is a fixed constant greater than $2$;
we conjecture that this is indeed the case, 
but we were unable to prove this.

\section{Parameterized Complexity} \label{sec:parameterized}
We have seen that the verification problem becomes easier if the
size of the candidate set $m=|P|$ is very small. We will now consider
other natural parameters of our problem, such as the number
of voters $n$ and the number of rounds $\ell$, and explore the complexity of the verification problem with respect to there parameters.

\begin{proposition}\label{prop:FPT_n}
    For each X $\in \{$(w-)JR, (w-)PJR, (w-)EJR$\}$, checking if an outcome provides X is FPT with respect to $n$.
\end{proposition}

The proof of the following proposition, as well as the results in Section~\ref{sec:mon}, 
make use of the following observation.

\begin{observation}\label{obs:compare}
    Let $X$ be a multiset of real numbers, and let $r\le |X|$ be a positive integer. Let $Y$ and $Z$ be two subsets of $X$  with $|Y|=|Z|=r$ such that 
    $Z$ consists of the $r$ smallest elements of $X$ (for some way of breaking ties). 
    Then for each $z\in Z$ we have $z\le \max_{y\in Y}y$.
\end{observation}

\begin{proposition}\label{prop:m-ell}
     For each X $\in \{$(w-)JR, (w-)PJR, (w-)EJR$\}$, 
    checking if an outcome provides X is FPT with respect to the combined parameter $(m,\ell)$ and XP with respect to $\ell$.
\end{proposition}
\begin{proof}
    Given an outcome $\mathbf o$, we proceed as follows.
	For each possible outcome ${\mathbf o}'$ and each subset $T\subseteq [\ell]$,
    compute the set 
        $N_{{\mathbf o}', T}=\{i\in N: o'_t\in s_{i, t}\text{ for all $t\in T$}\}$;
        this set consists of all voters who approve the outcome ${\mathbf o}'$ in
        each of the rounds in $T$. We then sort the voters in $N_{{\mathbf o}', T}$
	according to their satisfaction under $\mathbf o$ in non-decreasing order, 
	breaking ties lexicographically; 
 for each $r=1, \dots, |N_{{\mathbf o}', T}|$,
	let $N^r_{{\mathbf o}', T}$ 
	denote the set that consists of the first $r$ voters in this order. 

    To verify EJR, for each $r=1, \dots, |N_{{\mathbf o}', T}|$ we determine whether 
        there is an $i\in N^r_{{\mathbf o}', T}$ such that $\textit{sat}_i({\mathbf o})\ge \lfloor \frac{r}{n}\cdot |T| \rfloor$.
	We claim that $\mathbf o$ provides EJR if and only if it passes these checks 
        for all possible choices of ${\mathbf o}'$, $T$ and $r$. 

	Indeed, suppose $\mathbf o$ fails this check for some ${\mathbf o}'$, $T$ and $r$,
	and let  $N' = N^r_{{\mathbf o}', T}$. By construction, $\beta(N')\ge |T|$ and $|N'|=r$, so
	$\alpha(N')\ge \lfloor \frac{r}{n} \cdot |T| \rfloor$, i.e., $N'$ is a witness 
    that EJR is violated.
 
	Conversely, suppose that $\mathbf o$ fails EJR. 
    Then there is a set of voters $N'$ such that 
	$\max_{i\in N'}\textit{sat}_i({\mathbf o})<\alpha(N')$.
    By definition of $\beta(N')$, there exists a set $T'\subseteq [\ell]$, $|T'|=\beta(N')$,
	and an outcome ${\mathbf o}'$ such that for each $t\in T'$ and each $i\in N'$ it holds that
	$o'_t\in s_{i, t}$; note that this implies
	$N'\subseteq N_{{\mathbf o}', T'}$. Let $r=|N'|$, and consider the set
    $N_{{\mathbf o}', T'}^r$. Observe that $\beta(N_{{\mathbf o}', T'}^r)\ge |T'|=\beta(N')$ and $|N_{{\mathbf o}', T'}^r|=|N'|=r$, so 
    $\alpha(N_{{\mathbf o}', T'}^r)\ge \alpha(N')$.    
    On the other hand, by applying Observation~\ref{obs:compare}
    to the multisets $\{\mathit{sat}_i({\mathbf o}): i\in N'\}$ and
    $\{\mathit{sat}_i({\mathbf o}): i\in N_{{\mathbf o}', T'}^r\}$,  
    we conclude that for each $i\in N_{{\mathbf o}', T'}^r$ it holds that
    $\textit{sat}_i({\mathbf o})\le \max_{i\in N'} \textit{sat}_i({\mathbf o}) < \alpha(N')\le \alpha(N_{{\mathbf o}', T'}^r)$.
    That is, $N_{{\mathbf o}', T'}^r$ is also a witness that $\mathbf o$ fails to provide EJR. As our algorithm checks $N_{{\mathbf o}', T'}^r$, 
it will be able to detect that $\mathbf o$ fails to provide EJR.
 
    For PJR, instead of checking whether $\textit{sat}_i({\mathbf o})\ge \lfloor \frac{r}{n}\cdot |T| \rfloor$ for some $i\in N^r_{{\mathbf o}', T}$, we check that
    $|\{t\in [\ell]: o_t\in \cup_{i\in N^r_{{\mathbf o}', T}}s_{i, t}\}|\ge \lfloor \frac{r}{n}\cdot |T| \rfloor$, and for JR we check whether $\lfloor \frac{r}{n}\cdot |T| \rfloor \ge 1$ implies
    that $\textit{sat}_i({\mathbf o})\ge 1$ for some $i\in N^r_{{\mathbf o}', T}$.
    For w-JR, w-PJR and w-EJR, the algorithm can be simplified: instead of iterating through
    all $T\subseteq [\ell]$, it suffices to consider $T=[\ell]$. The rest of the argument goes through without change.
    
	The bound on the runtime follows, as we have $m^\ell$ possibilities
	for ${\mathbf o}'$ and $2^\ell$ possibilities for $T$, and we process
	each pair $({\mathbf o}', T)$ in polynomial time.
\end{proof}
Our  XP result with respect to $\ell$ is tight, at least for weak versions of the axioms: our next theorem shows that checking whether an outcome provides w-JR, w-PJR, or w-EJR is W[1]-hard with respect to $\ell$. This indicates that an FPT (in $\ell$) algorithm does not exist unless FPT = W[1].
\begin{theorem}\label{thm:W1-hard}
    For each of X $\in\{$w-JR, w-PJR, w-EJR$\}$, verifying whether an outcome provides X is {\em $W[1]$}-hard with respect to~$\ell$.
\end{theorem}

\section{Monotonic Preferences}\label{sec:mon}
Since the problem of checking whether an outcome provides (w-)JR, (w-)PJR, or 
(w-)EJR is intractable in general, it is natural to seek a restriction on voters' preferences that may yield positive results. In particular, 
a natural restriction in this context is monotonicity: once a voter
approves a candidate, she continues to approve it in subsequent rounds.
Formally, we will say that an election $(P, N, \ell, ({\mathbf s}_i)_{i\in N})$ 
is {\em monotonic} if 
for any two rounds $t,t' \in [\ell]$ with $t < t'$, each $p\in P$
and each $i\in N$ it holds that $p \in s_{i,t}$ implies $p \in s_{i,t'}$.
Monotonic elections occur if, e.g., candidates join the candidate pool 
over time, but never leave, and the voters' preferences over the available candidates
do not change.
Such preferences can also arise in settings where the candidates improve over time: e.g., job candidates become more experienced.
By reversing the time line, this can also model candidates that deteriorate over time (e.g., planning family meals from available ingredients).

We first note that verifying weak JR/PJR/EJR in monotonic elections is easy.

\begin{theorem}\label{thm:mon-ar}
    Given a monotonic election $(P, N, \ell, ({\mathbf s}_i)_{i\in N})$ togther with an outcome $\mathbf o$, for each of X $\in \{$w-JR, w-PJR, w-EJR$\}$ we can decide in polynomial time whether $\mathbf o$ provides X.
\end{theorem}

For strong notions of justified representation, we also obtain 
easiness-of-verification results, though the proof
requires an additional level of complexity.

\begin{theorem}\label{thm:struct}
    Given a monotonic election $(P, N, \ell, ({\mathbf s}_i)_{i\in N})$ together with an outcome $\mathbf o$, for each of X $\in \{$JR, PJR, EJR$\}$ we can decide in polynomial time whether $\mathbf o$ provides X.
\end{theorem}

\section{Finding EJR Outcomes} \label{sec:s-ejr}
So far, we focused on checking whether an outcome provides a representation guarantee. We will now switch gears, and explore the problem of finding a fair outcome.

\subsection{Greedy Cohesive Rule Provides EJR}\label{sec:algos}

Chandak et al.~\shortcite{chandak23} show that every temporal election admits
an outcome that provides EJR, by adapting the PAV rule~\cite{thiele1895pav}
and its local search variant ls-PAV~\cite{aziz2018complexityepjr} from the multiwinner setting to the temporal setting; the local search-based approach 
results in a rule that is polynomial-time computable.
However, they leave it as an open question whether another prominent voting rule, namely, the Greedy Cohesive Rule (GCR) \cite{bredereckF0N19,peters2021fjr} can be adapted to the temporal setting so as to provide EJR. 
    
Bulteau et al.~\shortcite{bulteau2021jrperpetual} describe an algorithm that is similar in spirit to GCR and constructs a PJR outcome; however, unlike GCR, the algorithm of Bulteau et al.~\shortcite{bulteau2021jrperpetual} proceeds in two stages. 
We will now describe a two-stage procedure 
that is inspired by the algorithm of Bulteau et al.~(but differs from it) and always finds EJR outcomes.

\begin{algorithm}
\caption{2-Stage Greedy Cohesive Rule.}
\label{alg:EJR}
\SetInd{0.8em}{0.3em}
{\bf Input}: Set of voters $N = \{1,\dots,n\}$, set of candidates $P = \{p_1,\dots,p_m\}$, number of rounds $\ell$, voters' approval sets $(\mathbf{s}_1,\dots,\mathbf{s}_n)$\;
$\mathbf{o} \leftarrow (p_1,\dots, p_1)$\;
$T \leftarrow [\ell]$\;
$\mathcal{V} := \{V \subseteq N: \beta(V)>0\}$\;
${\mathcal V}^+=\varnothing$\;

\While{$\mathcal{V}\neq\varnothing$}{
    Select a set $V\in\arg\max_{V\in{\mathcal V}}\alpha(V)$ with ties broken arbitrarily\;
    ${\mathcal V}^+\gets{\mathcal V}^+\cup\{V\}$\;
    \For{$V'\in \mathcal{V}$}{
        \If{$V\cap V'\neq\varnothing$}{
            $\mathcal{V} \leftarrow \mathcal{V} \setminus \{V'\}$\;
        }
    }
}
Sort ${\mathcal V}^+$ as $V_1, \dots, V_q$ so that $\beta(V_1)\le\dots\le\beta(V_q)$\;
\For{$k=1, \dots, q$}{
    Let $T'$ be a subset of 
                    $\{t\in T: \cap_{i\in V_k}s_{i, t}\neq\varnothing\}$ of size $\alpha(V_k)$\;
    $T \leftarrow T \setminus T'$\;
                \For{$t \in T'$}{
                    Pick $p\in \cap_{i\in V_k}s_{i, t}$ and set $o_t \leftarrow p$\;
                }
}

\Return{outcome $\mathbf{o}$}\;
\end{algorithm}

\begin{theorem} \label{thm:ejr_existence}
    Algorithm~\ref{alg:EJR} always outputs an outcome that provides EJR,
    and runs in time $O(2^n\cdot\textrm{poly}(n, m, \ell))$.
\end{theorem}
\begin{proof}
    The bound on the running time of the algorithm follows from the fact that 
    the size of $\mathcal V$ is $O(2^n)$, and the amount of computation performed 
    for each subset in $\mathcal V$ is polynomial in the input size.
    We now focus on correctness.

    Note that all sets in ${\mathcal V}^+$ are pairwise disjoint: for each $k\in [q]$, when we place $V_k$ in ${\mathcal V}^+$, we remove all sets that intersect
    $V_k$ from $\mathcal V$, and therefore no such set can be added to ${\mathcal V}^+$ at a future iteration.
    We will use this observation to argue that, while processing the sets in ${\mathcal V}^+$,
    for each set $V_k$ we consider there exist $\alpha(V_k)$ rounds in which all members of $V_k$ agree on at least one candidate.
    Suppose for a contradiction that this is not the case, and let $k \in [q]$ be the first index such that there are fewer than $\alpha(V_k)$ slots available out of the $\beta(V_k)$ rounds voters in $V_k$ agree on. This means that strictly more than $\beta(V_k) - \alpha(V_k)$ of these slots have been taken up in previous iterations, and hence $\sum_{r=1}^k\alpha(V_r) >\beta(V_k)$.
    Since $\beta(V_k)\ge \beta(V_r)$ for all $r\le k$, the inequality 
    implies
    \begin{equation}
        \sum_{r=1}^{k} \frac{|V_r|}{n} \cdot \beta(V_k) \geq
        \sum_{r=1}^{k} \frac{|V_r|}{n} \cdot \beta(V_r) \geq
        \sum_{r=1}^{k} \alpha(V_r) > \beta(V_k), 
    \end{equation}
    and hence $\sum_{r=1}^k|V_r|>n$, a contradiction
    with the fact the groups $V_1,\dots,V_{k-1}, V_k$ are pairwise
    disjoint, and $|V|=n$.

    We conclude that for each $V\in {\mathcal V}^+$ 
    it holds that $\textit{sat}_i({\mathbf o})\ge \alpha(V)$ for each $i\in V$.
    Hence, no set in ${\mathcal V}^+$ 
    can be a witness 
    that EJR is violated. Further, each set $V\in 2^N\setminus \mathcal V$
    has $\beta(V)=0$ and hence $\alpha(V)=0$, so
    it cannot witness that EJR is violated either. 
    It remains to consider sets in 
    ${\mathcal V}\setminus {\mathcal V}^+$. 
    
    Fix a set $V'\in {\mathcal V}\setminus {\mathcal V}^+$, 
    and suppose that it was deleted from $\mathcal V$ 
    when some $V$ was placed in ${\mathcal V}^+$ 
    (and hence  $\alpha(V')\le \alpha(V)$). 
    Moreover, when processing ${\mathcal V}^+$,  
    the algorithm ensured that the satisfaction of each voter $i\in V$ is at least 
    $\alpha(V)=\lfloor\frac{\beta(V)\cdot |V|}{n}\rfloor$. 
    As $V\cap V'\neq\varnothing$, 
    this means that the satisfaction of some voter $i\in V'$ is at least $\alpha(V)\ge \alpha(V')$.
    Thus, $V'$ cannot be a witness that EJR is violated.
\end{proof}

We note that for monotonic elections Algorithm~\ref{alg:EJR} can be modified
to run in polynomial time. Indeed, the proof of Theorem~\ref{thm:struct} shows that, when verifying representation axioms, it suffices to consider sets of the form $N_{p, t}^z$ (defined in the proof; see the appendix for details). If we modify 
Algorithm~\ref{alg:EJR} so that initially it only places sets of this form in $\mathcal V$, it will only have to consider $O(m\ell n)$ sets, each of them in polynomial time; on the other hand, by Theorem~\ref{thm:struct} its output would still provide EJR. 

While GCR runs in exponential time, it is a simple rule that can often be used to prove existence of fair outcomes; e.g., in the multiwinner setting, its outputs satisfy the stronger FJR axiom~\cite{peters2021fjr}, which is not satisfied by PAV or the Method of Equal Shares (MES). While FJR seems difficult to define for the temporal setting, this remains an interesting direction for future work, and GCR may prove to be a relevant tool. 
    Moreover, one may want to combine JR-like guarantees with additional welfare, coverage, or diversity guarantees; the associated problems are likely to be NP-hard, so GCR may be a useful starting point for an approximation or a fixed-parameter algorithm (that may be more practical than the ILP approach). 
    Finally, compared to PAV and MES, GCR may be easier to adapt to general monotone valuations (this is indeed the case in the multiwinner setting).

\subsection{An ILP For Finding EJR Outcomes} \label{sec:ilp}
We will now describe an algorithm for finding EJR outcomes 
that is based on integer linear programming (ILP). While
this algorithm does not run in polynomial time, it is very flexible:
e.g., we can easily modify it so as to find an EJR outcome that
maximizes the utilitarian social welfare, or provides satisfaction guarantees
to individual voters.

\begin{theorem}\label{thm:ilp}
    There exists an integer linear program (ILP) whose solutions correspond to outcomes 
    that provide EJR; the number of variables and the number of constraints of this ILP are bounded by a function of the number of voters $n$.
\end{theorem}

The following corollary illustrates the power of the ILP-based approach.

\begin{corollary}\label{cor:ejropt}
There is an FPT algorithm with respect to the number of voters that, given an election $E=(P, N, \ell, ({\mathbf s}_i)_{i\in N})$ and
a set of integers $\delta_1, \dots, \delta_n$, decides
whether there exists an EJR outcome of $E$ that guarantees
satisfaction $\delta_i$ to voter $i$ for each $i\in N$, and, if yes, finds an outcome
that maximizes the utilitarian social welfare among all outcomes with this 
property.
\end{corollary}
\begin{proof}
    The ILP in the proof of Theorem~\ref{thm:ilp} contains an expression
    encoding each voter's satisfaction. We can modify this ILP by adding constraints saying
    that the satisfaction of voter $i$ is at least $\delta_i$, and adding a goal function
    that maximizes the sum of voters' satisfactions. The resulting ILP admits
    an algorithm that is FPT with respect to $n$ by the classic result 
    of Lenstra~\shortcite{lenstra1983integer}.
\end{proof}

We note that, while Chandak et al.~\shortcite{chandak23} show that ls-PAV
rule can find EJR outcomes in polynomial time, their approach cannot
handle additional constraints on voters' welfare; hence 
Corollary~\ref{cor:ejropt} is not implied by their result.

\subsection{An Impossibility Result for EJR in the (Semi-)Online Setting}
In their analysis, Chandak et al.~\shortcite{chandak23} distinguish among (1) the \emph{online setting}, 
where the number of rounds $\ell$ is not known, voters' preferences are revealed round-by-round, and $o_t$ is selected
as soon as all $(s_{i, t})_{i\in N}$ are revealed, 
(2) the \emph{semi-online setting}, where $\ell$ is known, but preferences are revealed round-by-round and an outcome for a round needs to be chosen as soon as the preferences for that round have been revealed, and (3) the \emph{offline setting}, where we select
the entire outcome $\mathbf o$ given full access to $({\mathbf s}_i)_{i\in N}$.
The PAV rule and its local search variant, the GCR rule and the ILP approach only work in the offline setting. In contrast, 
Chandak et al.~\shortcite{chandak23} show that
a variant of MES satisfies w-EJR in the semi-online setting, but not in the online setting. They left open the existence of rules satisfying EJR in the semi-online setting.
Here, we resolve this open question and show that no rule can satisfy EJR in the semi-online setting (and hence in the online setting).
\begin{proposition} \label{prop:s-EJR_nonexistence}
    No rule satisfies EJR 
    in the semi-online setting.
\end{proposition}

\section{Conclusion}
We have explored the complexity of verifying whether a given outcome
of a temporal election satisfies one of the six representation axioms
considered by Chandak et al.~\shortcite{chandak23}.
We have obtained coNP-hardness results even for very restricted special cases of this problem: e.g., for strong versions of the axioms verification remains hard even if there are just two candidates.
We complement these hardness results with parameterized complexity results and a 
positive result for a structured setting, where candidates join the pool over time
and never leave. We also describe an ILP that can be used to find outcomes that provide EJR and satisfy additional constraints. Finally, we answer an open question of Chandak et al.~\shortcite{chandak23}, by showing
that a variant of the Greedy Cohesive Rule provides EJR.

Possible directions for future work (in addition to open problems listed in Section~\ref{sec:m}) include considering stronger variants of proportionality, such as FJR \cite{peters2021fjr} or core stability~\cite{aziz2017jr}, exploring other domain restrictions, and extending the temporal framework
to the broader setting of participatory budgeting~\cite{aziz2021pbordinal,lackner2021longtermpb,peters2021fjr}.

\newpage

\section*{Acknowledgments}
Edith Elkind was supported by the AI Programme of The Alan Turing Institute and an EPSRC Grant EP/X038548/1. Jannik Peters was supported by the Singapore Ministry of Education under grant number MOE-T2EP20221-0001.

\bibliography{abb,aaai25}

\begin{thebibliography}{47}
\providecommand{\natexlab}[1]{#1}

\bibitem[{Allouah et~al.(2023)Allouah, Kroer, Zhang, Avadhanula, Bohanon,
  Dania, Gocmen, Pupyrev, Shah, Stier-Moses, and
  Taarup}]{allouah2023fairallocationtime}
Allouah, A.; Kroer, C.; Zhang, X.; Avadhanula, V.; Bohanon, N.; Dania, A.;
  Gocmen, C.; Pupyrev, S.; Shah, P.; Stier-Moses, N.; and Taarup, K.~R. 2023.
\newblock Fair Allocation Over Time, with Applications to Content Moderation.
\newblock In \emph{Proceedings of the 29th ACM SIGKDD Conference on Knowledge
  Discovery and Data Mining (KDD)}, 25--35.

\bibitem[{Alouf-Heffetz et~al.(2022)Alouf-Heffetz, Bulteau, Elkind, Talmon, and
  Teh}]{alouf2022better}
Alouf-Heffetz, S.; Bulteau, L.; Elkind, E.; Talmon, N.; and Teh, N. 2022.
\newblock Better Collective Decisions via Uncertainty Reduction.
\newblock In \emph{Proceedings of the 31st International Joint Conference on
  Artificial Intelligence (IJCAI)}, 24--30.

\bibitem[{Aziz et~al.(2017)Aziz, Brill, Conitzer, Elkind, Freeman, and
  Walsh}]{aziz2017jr}
Aziz, H.; Brill, M.; Conitzer, V.; Elkind, E.; Freeman, R.; and Walsh, T. 2017.
\newblock Justified Representation in Approval-Based Committee Voting.
\newblock \emph{Social Choice and Welfare}, 48: 461--485.

\bibitem[{Aziz et~al.(2018)Aziz, Elkind, Huang, Lackner,
  S\'{a}nchez-Fern\'{a}ndez, and Skowron}]{aziz2018complexityepjr}
Aziz, H.; Elkind, E.; Huang, S.; Lackner, M.; S\'{a}nchez-Fern\'{a}ndez, L.;
  and Skowron, P. 2018.
\newblock On the Complexity of Extended and Proportional Justified
  Representation.
\newblock In \emph{Proceedings of the 32nd AAAI Conference on Artificial
  Intelligence (AAAI)}, 902--909.

\bibitem[{Aziz and Lee(2021)}]{aziz2021pbordinal}
Aziz, H.; and Lee, B. 2021.
\newblock Proportionally Representative Participatory Budgeting with Ordinal
  Preferences.
\newblock In \emph{Proceedings of the 35th AAAI Conference on Artificial
  Intelligence (AAAI)}, 5110--5118.

\bibitem[{Aziz and Lee(2020)}]{aziz2020expanding}
Aziz, H.; and Lee, B.~E. 2020.
\newblock The expanding approvals rule: Improving proportional representation
  and monotonicity.
\newblock \emph{Social Choice and Welfare}, 54(1): 1--45.

\bibitem[{Bampis, Escoffier, and Mladenovic(2018)}]{bampis2018fairtime}
Bampis, E.; Escoffier, B.; and Mladenovic, S. 2018.
\newblock Fair Resource Allocation Over Time.
\newblock In \emph{Proceedings of the 17th International Conference on
  Autonomous Agents and Multi-Agent Systems (AAMAS)}, 766--773.

\bibitem[{Bredereck et~al.(2018)Bredereck, Faliszewski, Igarashi, Lackner, and
  Skowron}]{bredereck2018diversity}
Bredereck, R.; Faliszewski, P.; Igarashi, A.; Lackner, M.; and Skowron, P.
  2018.
\newblock Multiwinner Elections with Diversity Constraints.
\newblock In \emph{Proceedings of the 32th AAAI Conference on Artificial
  Intelligence (AAAI)}, 933--940.

\bibitem[{Bredereck et~al.(2019)Bredereck, Faliszewski, Kaczmarczyk, and
  Niedermeier}]{bredereckF0N19}
Bredereck, R.; Faliszewski, P.; Kaczmarczyk, A.; and Niedermeier, R. 2019.
\newblock An Experimental View on Committees Providing Justified
  Representation.
\newblock In \emph{Proceedings of the 28th International Joint Conference on
  Artificial Intelligence (IJCAI)}, 109--115.

\bibitem[{Bredereck, Fluschnik, and
  Kaczmarczyk(2022)}]{bredereck2022committeechange}
Bredereck, R.; Fluschnik, T.; and Kaczmarczyk, A. 2022.
\newblock When Votes Change and Committees Should (Not).
\newblock In \emph{Proceedings of the 31st International Joint Conference on
  Artificial Intelligence (IJCAI)}, 144--150.

\bibitem[{Bredereck, Kaczmarczyk, and
  Niedermeier(2020)}]{bredereck2020successivecommittee}
Bredereck, R.; Kaczmarczyk, A.; and Niedermeier, R. 2020.
\newblock Electing Successive Committees: Complexity and Algorithms.
\newblock In \emph{Proceedings of the 34th AAAI Conference on Artificial
  Intelligence (AAAI)}, 1846--1853.

\bibitem[{Brill et~al.(2024)Brill, Gölz, Peters, Schmidt-Kraepelin, and
  Wilker}]{brill2024partyapportionment}
Brill, M.; Gölz, P.; Peters, D.; Schmidt-Kraepelin, U.; and Wilker, K. 2024.
\newblock Approval-based apportionment.
\newblock \emph{Mathematical Programming}, 203: 77--105.

\bibitem[{Brill and Peters(2023)}]{Brill023}
Brill, M.; and Peters, J. 2023.
\newblock Robust and Verifiable Proportionality Axioms for Multiwinner Voting.
\newblock In \emph{Proceedings of the 24th ACM Conference on Economics and
  Computation (EC)}, 301.

\bibitem[{Bulteau et~al.(2021)Bulteau, Hazon, Page, Rosenfeld, and
  Talmon}]{bulteau2021jrperpetual}
Bulteau, L.; Hazon, N.; Page, R.; Rosenfeld, A.; and Talmon, N. 2021.
\newblock Justified Representation for Perpetual Voting.
\newblock \emph{IEEE Access}, 9: 96598--96612.

\bibitem[{Celis, Huang, and Vishnoi(2018)}]{celis2018fairness}
Celis, L.~E.; Huang, L.; and Vishnoi, N.~K. 2018.
\newblock Multiwinner Voting with Fairness Constraints.
\newblock In \emph{Proceedings of the 27th International Joint Conference on
  Artificial Intelligence (IJCAI)}, 144--151.

\bibitem[{Chandak, Goel, and Peters(2024)}]{chandak23}
Chandak, N.; Goel, S.; and Peters, D. 2024.
\newblock Proportional Aggregation of Preferences for Sequential Decision
  Making.
\newblock In \emph{Proceedings of the 38th AAAI Conference on Artificial
  Intelligence (AAAI)}, 9573--9581.

\bibitem[{Conitzer, Freeman, and Shah(2017)}]{conitzer2017fairpublic}
Conitzer, V.; Freeman, R.; and Shah, N. 2017.
\newblock Fair Public Decision Making.
\newblock In \emph{Proceedings of the 18th ACM Conference on Economics and
  Computation (EC)}, 629--646.

\bibitem[{Cygan et~al.(2015)Cygan, Fomin, Kowalik, Lokshtanov, Marx, Pilipczuk,
  Pilipczuk, and Saurabh}]{cygan2015parameterized}
Cygan, M.; Fomin, F.~V.; Kowalik, {\L}.; Lokshtanov, D.; Marx, D.; Pilipczuk,
  M.; Pilipczuk, M.; and Saurabh, S. 2015.
\newblock \emph{Parameterized algorithms}, volume~5.
\newblock Springer.

\bibitem[{Delemazure et~al.(2023)Delemazure, Demeulemeester, Eberl, Israel, and
  Lederer}]{delemazure2022spelection}
Delemazure, T.; Demeulemeester, T.; Eberl, M.; Israel, J.; and Lederer, P.
  2023.
\newblock Strategyproofness and Proportionality in Party-Approval Multiwinner
  Elections.
\newblock In \emph{Proceedings of the 37th AAAI Conference on Artificial
  Intelligence (AAAI)}, 5591--5599.

\bibitem[{Do et~al.(2022)Do, Hervouin, Lang, and
  Skowron}]{do2022onlinecommittee}
Do, V.; Hervouin, M.; Lang, J.; and Skowron, P. 2022.
\newblock Online Approval Committee Elections.
\newblock In \emph{Proceedings of the 31st International Joint Conference on
  Artificial Intelligence (IJCAI)}, 251--257.

\bibitem[{Elkind et~al.(2017)Elkind, Faliszewski, Skowron, and
  Slinko}]{elkind2017propertiesmwv}
Elkind, E.; Faliszewski, P.; Skowron, P.; and Slinko, A. 2017.
\newblock Properties of multiwinner voting rules.
\newblock \emph{Social Choice and Welfare}, 48: 599--632.

\bibitem[{Elkind, Kraiczy, and Teh(2022)}]{elkind2022temporalslot}
Elkind, E.; Kraiczy, S.; and Teh, N. 2022.
\newblock Fairness in Temporal Slot Assignment.
\newblock In \emph{Proceedings of the 15th International Symposium on
  Algorithmic Game Theory (SAGT)}, 490--507.

\bibitem[{Elkind et~al.(2025)Elkind, Lam, Latifian, Neoh, and
  Teh}]{elkind2024temporalfairdivision}
Elkind, E.; Lam, A.; Latifian, M.; Neoh, T.~Y.; and Teh, N. 2025.
\newblock Temporal Fair Division of Indivisible Items.
\newblock In \emph{Proceedings of the 24th International Conference on
  Autonomous Agents and Multi-Agent Systems (AAMAS)}.

\bibitem[{Elkind, Neoh, and Teh(2024)}]{elkind2024temporalelections}
Elkind, E.; Neoh, T.~Y.; and Teh, N. 2024.
\newblock Temporal Elections: Welfare, Strategyproofness, and Proportionality.
\newblock In \emph{Proceedings of the 27th European Conference on Artificial
  Intelligence (ECAI)}, 3292--3299.

\bibitem[{Elkind, Obraztsova, and Teh(2024)}]{elkind2024temporal}
Elkind, E.; Obraztsova, S.; and Teh, N. 2024.
\newblock Temporal Fairness in Multiwinner Voting.
\newblock In \emph{Proceedings of the 38th AAAI Conference on Artificial
  Intelligence (AAAI)}, 22633--22640.

\bibitem[{Fain, Munagala, and Shah(2018)}]{fain2018publicgoods}
Fain, B.; Munagala, K.; and Shah, N. 2018.
\newblock Fair Allocation of Indivisible Public Goods.
\newblock In \emph{Proceedings of the 19th ACM Conference on Economics and
  Computation (EC)}, 575--592.

\bibitem[{Faliszewski et~al.(2017)Faliszewski, Skowron, Slinko, and
  Talmon}]{faliszewski2017mwv}
Faliszewski, P.; Skowron, P.; Slinko, A.; and Talmon, N. 2017.
\newblock Multiwinner Voting: A New Challenge for Social Choice Theory.
\newblock In \emph{Trends in Computational Social Choice}, chapter~2, 27--47.

\bibitem[{Freeman, Zahedi, and Conitzer(2017)}]{freeman2017dynamicsocialchoice}
Freeman, R.; Zahedi, S.~M.; and Conitzer, V. 2017.
\newblock Fair and Efficient Social Choice in Dynamic Settings.
\newblock In \emph{Proceedings of the 26th International Joint Conference on
  Artificial Intelligence (IJCAI)}, 4580--4587.

\bibitem[{Garey, Johnson, and Stockmeyer(1976)}]{garey1976npcompletegraphs}
Garey, M.~R.; Johnson, D.~S.; and Stockmeyer, L.~J. 1976.
\newblock Some simplified NP-complete graph problems.
\newblock \emph{Theoretical Computer Science}, 1(3): 237--267.

\bibitem[{Lackner(2020)}]{lackner2020perpetual}
Lackner, M. 2020.
\newblock Perpetual Voting: {F}airness in Long-Term Decision Making.
\newblock In \emph{Proceedings of the 34th AAAI Conference on Artificial
  Intelligence (AAAI)}, 2103--2110.

\bibitem[{Lackner and Maly(2023)}]{lackner2023proportionalPV}
Lackner, M.; and Maly, J. 2023.
\newblock Proportional Decisions in Perpetual Voting.
\newblock In \emph{Proceedings of the 37th AAAI Conference on Artificial
  Intelligence (AAAI)}, 5722--5729.

\bibitem[{Lackner, Maly, and Nardi(2023)}]{lackner2023freeriding}
Lackner, M.; Maly, J.; and Nardi, O. 2023.
\newblock Free-Riding in Multi-Issue Decisions.
\newblock In \emph{Proceedings of the 22nd International Conference on
  Autonomous Agents and Multi-Agent Systems (AAMAS)}, 2040--2048.

\bibitem[{Lackner, Maly, and Rey(2021)}]{lackner2021longtermpb}
Lackner, M.; Maly, J.; and Rey, S. 2021.
\newblock Fairness in Long-Term Participatory Budgeting.
\newblock In \emph{Proceedings of the 30th International Joint Conference on
  Artificial Intelligence (IJCAI)}, 299--305.

\bibitem[{Lackner and Skowron(2023)}]{lackner2022abc}
Lackner, M.; and Skowron, P. 2023.
\newblock \emph{Multi-Winner Voting with Approval Preferences}.
\newblock Springer.

\bibitem[{Lenstra~Jr(1983)}]{lenstra1983integer}
Lenstra~Jr, H.~W. 1983.
\newblock Integer programming with a fixed number of variables.
\newblock \emph{Mathematics of Operations Research}, 8(4): 538--548.

\bibitem[{Masa{\v r}{\'\i}k, Pierczy{\'n}ski, and Skowron(2024)}]{MPS23a}
Masa{\v r}{\'\i}k, T.; Pierczy{\'n}ski, G.; and Skowron, P. 2024.
\newblock A Generalised Theory of Proportionality in Collective Decision
  Making.
\newblock In \emph{Proceedings of the 25th ACM Conference on Economics and
  Computation (EC)}, 734--754.

\bibitem[{Neoh and Teh(2025)}]{neoh2025strategicchores}
Neoh, T.~Y.; and Teh, N. 2025.
\newblock Strategic Manipulation in Temporal Voting with Undesirable Candidates
  (Student Abstract).
\newblock In \emph{Proceedings of the 39th AAAI Conference on Artificial
  Intelligence (AAAI)}.

\bibitem[{Page, Shapiro, and Talmon(2022)}]{page2021electing}
Page, R.; Shapiro, E.; and Talmon, N. 2022.
\newblock Electing the Executive Branch.
\newblock \emph{arXiv preprint arXiv:2009.09734}.

\bibitem[{Parkes and Procaccia(2013)}]{parkes2013dynamicsocialchoice}
Parkes, D.; and Procaccia, A. 2013.
\newblock Dynamic Social Choice with Evolving Preferences.
\newblock In \emph{Proceedings of the 27th AAAI Conference on Artificial
  Intelligence (AAAI)}, 767--773.

\bibitem[{Patro et~al.(2022)Patro, Jana, Chakraborty, Gummadi, and
  Ganguly}]{patro2022virtualconf}
Patro, G.~K.; Jana, P.; Chakraborty, A.; Gummadi, K.~P.; and Ganguly, N. 2022.
\newblock Scheduling Virtual Conferences Fairly: Achieving Equitable
  Participant and Speaker Satisfaction.
\newblock In \emph{Proceedings of the ACM Web Conference (WWW) 2022},
  2646--2656.

\bibitem[{Peeters(2003)}]{peeters}
Peeters, R. 2003.
\newblock The maximum edge biclique problem is {NP}-complete.
\newblock \emph{Discrete Applied Mathematics}, 131(3): 651--654.

\bibitem[{Peters, Pierczy\'{n}ski, and Skowron(2021)}]{peters2021fjr}
Peters, D.; Pierczy\'{n}ski, G.; and Skowron, P. 2021.
\newblock Proportional Participatory Budgeting with Additive Utilities.
\newblock In \emph{Advances in Neural Information Processing Systems
  (NeurIPS)}, volume~34, 12726--12737.

\bibitem[{Peters and Skowron(2020)}]{peters2020proportionality}
Peters, D.; and Skowron, P. 2020.
\newblock Proportionality and the limits of welfarism.
\newblock In \emph{Proceedings of the 21st ACM Conference on Economics and
  Computation (EC)}, 793--794.

\bibitem[{S\'{a}nchez-Fern\'{a}ndez et~al.(2017)S\'{a}nchez-Fern\'{a}ndez,
  Elkind, Lackner, Fern\'{a}ndez, Fisteus, Basanta~Val, and
  Skowron}]{sf2017pjr}
S\'{a}nchez-Fern\'{a}ndez, L.; Elkind, E.; Lackner, M.; Fern\'{a}ndez, N.;
  Fisteus, J.; Basanta~Val, P.; and Skowron, P. 2017.
\newblock Proportional Justified Representation.
\newblock In \emph{Proceedings of the 31st AAAI Conference on Artificial
  Intelligence (AAAI)}, 670--676.

\bibitem[{Skowron and G\'{o}recki(2022)}]{skowron2022proppublic}
Skowron, P.; and G\'{o}recki, A. 2022.
\newblock Proportional Public Decisions.
\newblock In \emph{Proceedings of the 36th AAAI Conference on Artificial
  Intelligence (AAAI)}, 5191--5198.

\bibitem[{Thiele(1895)}]{thiele1895pav}
Thiele, T. 1895.
\newblock Om Flerfoldsvalg.
\newblock \emph{Oversigt Over Det Kongelige Danske Videnskabernes Selskabs
  Forhandlinger}, 415--441.

\bibitem[{Zech et~al.(2024)Zech, Boehmer, Elkind, and
  Teh}]{zech2024multiwinnerchange}
Zech, V.; Boehmer, N.; Elkind, E.; and Teh, N. 2024.
\newblock Multiwinner Temporal Voting with Aversion to Change.
\newblock In \emph{Proceedings of the 27th European Conference on Artificial
  Intelligence (ECAI)}, 3236--3243.

\end{thebibliography}

\newpage

\appendix

\onecolumn

\begin{center}
\Large
\textbf{Appendix}
\end{center}

\vspace{2mm}
\section{Omitted Proofs from Section \ref{sec:hard}}
\subsection{Proof of Theorem \ref{thm:np-hard-ap} (continued)}
We will now prove that verifying whether an outcome provides w-EJR is coNP-hard even if $|P|=2$.
    To this end, we reduce from the NP-hard problem \textsc{Independent Set with Maximum Degree of Three (IS-3)} \cite{garey1976npcompletegraphs}.
    An instance of \textsc{IS-3} is given by a graph $G = (V,E)$ with a vertex set $V$ and an edge set $E$ such that every vertex in $V$ has at most three neighbors, together with a parameter $\kappa$; it is a `yes'-instance if there exists a subset of vertices $V' \subseteq V$ with $|V'| \geq \kappa$ such that for all vertices $u,v \in V'$ it holds that $\{u,v\} \notin E$, and a `no' instance otherwise. We can assume without loss of generality that $\kappa\le |V|-3$: indeed, if $\kappa\ge |V|-2$, we can solve \textsc{IS-3} in time $O(|V|^2)$ by checking all size-$\kappa$ subsets of $V$.
    
    Given an instance $(G,\kappa)$ of \textsc{IS-3} with $G = (V,E)$, where $V = \{v_1,\dots, v_\nu\}$ and $\kappa\le \nu-3$, we construct an election with a set of voters $N = \{1,\dots,\nu\} \cup N_0$, where $|N_0| = \nu + 1 - \kappa$, a set of candidates $P = \{p,q\}$, and $\ell = 2\nu + 1 - \kappa$ rounds. 
    For each $i \in N$, the preferences of voter $i$ in round $t$ are given by: 
    \begin{equation*}
    	s_{i,t} = 
    	\begin{cases}
            \{p\} &\text{if } t >\nu\\
    		\{q\} & \text{if } i\in [\nu] \text{ and } t = i; \\
    		\{p\} & \text{if } i, t\in [\nu] \text{ and }\{v_i,v_t\} \in E; \\
    		\{p,q\} & \text{in all other cases}.
    	\end{cases}
    \end{equation*}

    Set $\mathbf{o} = (q,\dots,q)$.
    We claim that there exists a set of voters $N^* \subseteq [\nu] \cup N_0$ such that $\beta(N^*) = \ell$ but $\textit{sat}_i(\mathbf{o}) < \alpha(N^*)$ for all $i \in N^*$ if and only if the set $V'=\{v_i: i\in N^*\cap [\nu]\}$
    is an independent set of size at least $\kappa$ in $G$.
    
    For the `if' direction, let $V'$ be an independent set of size at least $\kappa$ in $G$, and let $N'$ be the set of voters that correspond to the vertices in $V'$. Set $N^*=N'\cup N_0$.
    We claim that $\beta(N^*)=\ell$.
    Indeed, consider two voters $i,i' \in N'$, a voter $j \in N_0$, and a round $t\in [\ell]$;
    we will now show that $s_{i, t}\cap s_{i', t}\cap s_{j, t}\neq\varnothing$. To this end, we consider the following cases:
    \begin{itemize}
        \item $t>\nu$. Then $s_{i, t}=s_{i', t}=s_{j, t} = \{p\}$.
        \item $t=i$. Then $s_{i,t} = \{q\}$ and $s_{i',t} = s_{j,t} = \{p,q\}$;
        \item $t=i'$. Then $s_{i',t} = \{q\}$ and $s_{i,t} = s_{j,t} = \{p,q\}$;
        \item $t\in [\nu]\setminus\{i, i'\}$. Then $p\in s_{i,t}$, $p\in s_{i',t}$, $s_{j,t} = \{p,q\}$.
    \end{itemize}
    As all three voters approve a common candidate in each round, we have $\beta(N^*) = \ell$.
   Then, we get 
    \begin{align*}
        \alpha(N^*) & = \left\lfloor \beta(N^*) \cdot \frac{|N^*|}{2\nu+1-\kappa} \right\rfloor 
        \ge \left\lfloor (2\nu + 1 - \kappa) \cdot \frac{\nu + 1 }{2\nu+1-\kappa}  \right\rfloor 
        = \nu + 1.
    \end{align*}
    However, since $s_{i,t} = \{p\}$ for all $i \in N$ and all $t > \nu$, the satisfaction of every voter in $N^*$ from outcome $\mathbf o$ does not exceed $\nu$.
    Thus, we obtain $\textit{sat}_i(\mathbf{o}) < \alpha(N^*)$ for all $i \in N^*$.

    For the `only if' direction, suppose there exists a set of voters $N^* \subseteq [\nu] \cup N_0$ such that $\beta(N^*) = \ell$ but $\textit{sat}_i(\mathbf{o}) < \alpha(N^*)$ for all $i \in N^*$.
    Let $N'=N^*\cap [\nu]$, and let $V'=\{v_i: i\in N'\}$. We claim that $V'$ is an independent set. Indeed, if $\{v_i, v_j\}\in E$ for some $i, j\in N'$ then
    $s_{i,i} = \{q\}$ and $s_{i,j} = \{p\}$, which is a contradiction with
    $\beta(N^*) = \ell$. It remains to argue that $|V'|\ge \kappa$. To this end, we note that $n=\ell=2\nu+1-\kappa$ and $\beta(N^*)=\ell$ implies that $\alpha(N^*)=|N^*|$, 
    and consider two cases.
    \begin{itemize}
        \item $N^*\subseteq [\nu]$. In this case, $N'=N^*$.
        Since every vertex in $V$ has at most three neighbors, for each $i \in N'$ we have  $\textit{sat}_i(\mathbf{o}) \geq \nu-3$.
        On the other hand, by our assumption, $\textit{sat}_i(\mathbf{o}) < \alpha(N^*)=|N^*|$ for all $i \in N^*$.
        Therefore, $V'$ is an independent set of size $|N^*|>\nu-3\ge \kappa$.
         \item $N^*\cap N_0\neq\varnothing$. Then, consider an agent $i\in N^*\cap N_0$. For each
    $i \in N_0$ we have $\mathit{sat}_i(\mathbf{o}) = \nu$ and hence $|N^*|=\alpha(N^*)>\nu$. Since $|N_0|=\nu+1-\kappa$, 
    this implies $|V'|=|N'|=|N^*\setminus N_0|> \kappa-1$. Thus, $|V'|$ is an independent set of size at least $\kappa$.
    \end{itemize}
    In both cases, we conclude that $V'$ is an independent set of size at least $\kappa$, and hence our input instance of \textsc{IS-3} was a `yes'-instance.   
    This completes the proof.

\subsection{Proof of Theorem \ref{thm:np-hard}}
It is easy to see that our problems are in coNP: 
we can modify the proof of Theorem~\ref{thm:np-hard-ap}
by omitting the condition $\beta(N')=\ell$.

To establish coNP-hardness, 
we give a reduction from the NP-hard problem {\sc Maximum Edge Biclique}~\cite{peeters}.
An instance of {\sc Maximum Edge Biclique} is given by a bipartite graph $G=(L\cup R, E)$
and a parameter $\kappa$; it is a `yes'-instance if and only if $G$ contains a biclique
with at least $\kappa$ edges, i.e., there exist $L'\subseteq L$, 
$R'\subseteq R$ such that for each $i\in L'$, $j\in R'$ it holds that $\{i, j\}\in E$ and
$|L'|\cdot |R'|\ge \kappa$. 

We can assume without loss of generality that $\kappa>|L|+|R|$. Indeed, if this is not the case, 
we can construct a new graph $G'$ by
replacing each vertex $v\in L\cup R$ with $\xi=|L|+|R|+1$ copies $v_1, \dots, v_\xi$
and adding an edge between the $i$-th copy of $v$ and the $j$-th copy of $u$ for all $i,j \in [\xi] \times [\xi]$
if and only if $\{v, u\}\in E$. The graph $G'$ has $\xi(|L|+|R|)$ vertices, 
and it contains a biclique with $\xi^2\kappa$ edges if and only if $G$ contains
a biclique with $\kappa$ edges; our choice of $\xi$ ensures that
$\xi^2\kappa>\xi(|L|+|R|)$.

We will give a coNP-hardness proof for JR and then extend it to PJR and EJR.
Consider an instance of {\sc Maximum Edge Biclique} given by a graph $G=(L\cup R, E)$
and a parameter $\kappa$, where $L=\{v_1, \dots, v_\nu\}$, $R=\{u_1, \dots, u_\lambda\}$
and $\kappa>\nu+\lambda$. 
We construct an election with a set of voters $N = \{1, \dots, \nu\}\cup N_0$,
where $|N_0|=\kappa-\nu$,  
a set of candidates $P=\{p, q\}$, and $\ell = \lambda$ rounds. 
For each $i \in [\nu]$ and each $t\in [\ell]$ we set $s_{i, t}=\{p\}$
if $\{v_i, u_t\}\in E$ and $s_{i, t}=\varnothing$ otherwise;
for each $i\in N_0$ and each $t\in [\ell]$ we set $s_{i, t}=\{q\}$.
Let ${\mathbf o}=(q, \dots, q)$.

Suppose $G$ contains a biclique with parts $L'\subseteq L$ and $R'\subseteq R$
and at least $\kappa$ edges, and consider the set of voters $N'=\{i: v_i\in L'\}$
and the set of rounds $T=\{t\in [\ell]: u_t\in R'\}$. Note that 
for each $i\in N'$, $t\in T$ we have $s_{i, t}=\{p\}$ and hence $\beta(N')\ge |T|=|R'|$.
Therefore, we obtain 
$$
\alpha(N')= \left\lfloor\beta(N')\cdot\frac{|N'|}{|N|}\right\rfloor\ge 
\left\lfloor\frac{|R'|\cdot|L'|}{|N|}\right\rfloor\ge \frac{\kappa}{\kappa}=1.
$$
Thus, $N'$ has positive demand, 
but no voter in $N'$ derives positive satisfaction from $\mathbf o$, i.e., JR is violated.

Conversely, suppose that there is a group of voters $N'$ with positive demand
such that $\textit{sat}_i({\mathbf o})=0$ for each $i\in N'$.
Then $N'\subseteq [\nu]$, as all voters in $N_0$ approve $q$ in each round.
Consider a set of rounds $T\subseteq [\ell]$ given by
$T= \{t\in [\ell]: \cap_{i\in N'}s_{i, t}\neq\varnothing\}$; we have $s_{i, t}=\{p\}$
for each $i\in N'$, $t\in T$. Thus, $(N', T)$ corresponds to a biclique in $G$:
indeed, we have $s_{i, t}=\{p\}$ if and only if $\{v_i, u_t\}\in E$.
It follows that $G$ contains a biclique with at least $|N'|\cdot |T|$ edges. 
On the other hand, we have $\beta(N')=|T|$ and hence 
$\alpha(N')=\lfloor\frac{|N'|\cdot|T|}{\kappa}\rfloor$. Thus, 
$\alpha(N')\ge 1$ implies $|N'|\cdot |T|\ge \kappa$.

To extend our result to PJR and EJR, we note that if $\mathbf o$ provides PJR (respectively, 
EJR) then it also provides JR, and hence, as argued above, $G$ does not contain a biclique
of size $\kappa$. Conversely, if $\mathbf o$ violates PJR (respectively, EJR), there exists a group
of voters $N'$ with a positive demand such that 
$|\{t\in [\ell]: o_t\in \cup_{i\in N'}s_{i, t}\}|<\alpha(N')$
(respectively, $\textit{sat}_i({\mathbf o})<\alpha(N')$ for all $i\in N'$). Now, the satisfaction of each 
voter in $N_0$ is $\ell$, and the demand of any group of voters is at most $\ell$.
Hence, $N'\cap N_0=\varnothing$, i.e., $N'\subseteq [\nu]$. 
As argued above, the conditions $N'\subseteq [\nu]$ and $\alpha(N')>0$
imply that the pair $(N', T)$, where 
$T= \{t\in [\ell]: \cap_{i\in N'}s_{i, t}\neq\varnothing\}$, 
corresponds to a biclique
in $G$ that has at least $\kappa$ edges.

\section{Omitted Proofs from Section \ref{sec:parameterized}}
\subsection{Proof of Proposition \ref{prop:FPT_n}}
Fix an outcome $\mathbf o$.
    We go through all subsets of voters $N'\subseteq N$. For each subset
    $N'$ we compute $\beta(N')$, by going through all rounds
    $t\in [\ell]$ and, for each round $t$, checking whether 
    $\cap_{i\in N'} s_{i, t}\neq \varnothing$.
    We then compute $\alpha(N')$ as well as the satisfaction that each voter in $N'$
    derives from $\mathbf o$. We can then easily verify whether the (w-)EJR
    (respectively, (w-)PJR, (w-)JR) condition is violated for $N'$.

    The bound on the running time follows immediately: we consider $2^n$ subsets
    of voters and perform a polynomial amount of work for each subset.

\subsection{Proof of Theorem \ref{thm:W1-hard}}
    To show W[1]-hardness, we reduce from the \emph{multicolored clique} problem. In this problem, we are given a parameter $k$ and a $k$-partite graph $G = (V_1 \sqcup \dots \sqcup V_k, E)$ with the goal being to find a clique of size $k$. It is well known that multicolored clique is W[1]-hard when parameterized by $k$ \citep{cygan2015parameterized}. We construct an election as follows: Let $\lvert V_1 \sqcup \dots \sqcup V_k\rvert = n'$ be the number of vertices in $G$. 
    
    In our election instance, we have $k$ rounds, one for each partition of $G$. At round $i \in [k]$, the candidates are the vertices in $V_i$ together with one dummy candidate. 
    Let each voter corresponding to a vertex in $V_i$ only approve of the candidate corresponding to itself. 
    Let each voter corresponding to a vertex not in $V_i$ approve all candidates corresponding to vertices connected to it. 
    The dummy candidate is approved by no voter.
    
    Finally, to make the instance balanced, we compare $\frac{n'}{k}$ to $k$. 
    If $\frac{n'}{k} < k$, we add $k^2 - n$ voters that approves of no candidates to the instance. 
    Thus, the $k$ vertices corresponding to a clique would correspond to exactly $\frac{n}{k}$ many voters. 
    If $\frac{n'}{k} > k$ we add $n'' = \lceil\frac{n' - k^2}{(k-1)}\rceil$ voters that each approve of every candidate to the instance. 
    Therefore, the $n''$ new voters together with $k$ potential clique voters correspond to an almost $\frac{n' + n''}{k}$ fraction of the new instance. 
    
    We now claim that the committee consisting of dummy candidates does not satisfy w-JR if and only if there exists a multicolored clique.
    
    First, assume that committee consisting of dummy candidates does not satisfy w-JR. Then there must exist $\frac{n}{k} \ge k$ vertices agreeing in every instance. Based on the way we added candidates, this must include at least $k$ voters corresponding to vertices. These $k$ voters must come from different partitions of the graph, as otherwise they would not agree in at least one round. Further, in the other rounds, they must agree on a vertex and are thus all connected to this vertex. Therefore, they form a clique. 
    Analogously, all the voters correspond to a clique (together with potentially added voters approving everything) also witness a w-JR violation, since they all agree on a candidate during each round and are large enough by construction.
    
\section{Omitted Proofs from Section \ref{sec:mon}}
\subsection{Proof of Theorem \ref{thm:mon-ar}}
Suppose that $(P, N, \ell, ({\mathbf s}_i)_{i\in N})$ is monotonic, 
and consider a group of voters $N'\subseteq N$.
Note that $\beta(N')=\ell$ if and only if 
$\cap_{i\in N'}s_{i, 1}\neq\varnothing$: if all voters in $N'$
approve some candidate $p$ in the first round, they also approve $p$
in all subsequent rounds. This enables us to proceed similarly
to the proof of Proposition~\ref{prop:m-ell}.

Specifically, fix an outcome $\mathbf o$, 
and for each $p\in P$, let $N_p=\{i\in N: p\in s_{i, 1}\}$.
Order the voters in $N_p$ according to their satisfaction from $\mathbf o$
in non-decreasing order, breaking ties lexicographically. 
For each $z=1, \dots, |N_p|$, let $N_p^z$
be the set that consists of the first $z$ voters in this order. 
Note that $\beta(N_p^z)=\ell$.

We claim that $\mathbf o$ provides w-EJR if and only if 
for each $p\in P$ and each $z=1, \dots, |N_p|$ it holds that
$\textit{sat}_i({\mathbf o})\ge \alpha(N_p^z)$ for some $i\in N_p^z$.
Note that this means that we can verify whether $\mathbf o$ 
provides w-EJR by considering at most $mn$ groups of voters and performing 
a polynomial amount of computation for each group, i.e., in polynomial time.

Now, if there exists a $p\in P$ and  $z\in\{1, \dots, |N_p|\}$
such that $\textit{sat}_i({\mathbf o})< \alpha(N_p^z)$ for all 
$i\in N_p^z$, then $\mathbf o$ fails to provide w-EJR, as witnessed by $N_p^z$.

Conversely, if there is a set $N'$ witnessing that $\mathbf o$
fails to provide w-EJR, we have $\beta(N')=\ell$ and hence
there exists a candidate $p\in P$ such that $p\in \cap_{i\in N'} s_{i, 1}$;
in particular, this implies $N'\subseteq N_p$. 
Let $z=|N'|$, and let $u = \max_{i\in N'}\textit{sat}_i({\mathbf o})$;
since $N'$ witnesses that $\mathbf o$ fails to provide w-EJR,
we have $u< \alpha(N')$.
Since $\beta(N')= \beta(N^z_p)$ and $|N'|=z=|N_p^z|$, 
we have $\alpha(N')=\alpha(N_p^z)$. 
Also, by Observation~\ref{obs:compare} (applied to the satisfactions 
of voters in $N'$ and $N_p^z$), for each $j\in N_p^z$ we have 
$\textit{sat}_j({\mathbf o})\le u < \alpha(N')=\alpha(N_p^z)$, i.e., 
$N_p^z$ is also a witness that $\mathbf o$ fails to provide w-EJR.
As our algorithm checks $N_p^z$, 
it will be able to detect that $\mathbf o$ fails to provide w-EJR.

To verify whether $\mathbf o$ provides w-PJR or w-JR, we modify
the checks that the algorithm performs for each group of voters, 
just as in the proof of Proposition~\ref{prop:m-ell}; 
we omit the details.

\subsection{Proof of Theorem \ref{thm:struct}}
Fix a monotonic election $(P, N, \ell, ({\mathbf s}_i)_{i\in N})$ 
    and an outcome $\mathbf o$.
    We can assume without loss of generality that $P=\cup_{i\in N}s_{i, \ell}$:
    if a candidate receives no approvals in the last round, it receives 
    no approvals at all, and can be removed.

    Now, for each $p\in P$ and each $t\in [\ell]$
    let $N_{p, t}=\{i\in N: p\in s_{i, t}\}$.
    Note that all voters in $N_{p, t}$ approve $p$ in each of the rounds
    $t, t+1, \dots, \ell$, so $\beta(N_{p, t})\ge \ell-t+1$.
    Order the voters in $N_{p, t}$ according to their satisfaction from $\mathbf o$
    in non-decreasing order, breaking ties lexicographically. 
    For each $z=1, \dots, |N_{p, t}|$, let $N_{p, t}^z$ be the set that consists 
    of the first $z$ voters in this order.

    We claim that $\mathbf o$ provides EJR if and only if 
    for each $p\in P$, $t\in [\ell]$ and each $z=1, \dots, |N_{p, t}|$ 
    it holds that $\textit{sat}_i({\mathbf o})\ge \alpha(N_{p, t}^z)$ 
    for some $i\in N_{p, t}^z$. 
    This means that we can verify whether $\mathbf o$ provides EJR
    by considering at most $m\ell n$ groups of voters and performing 
    a polynomial amount of computation for each group, i.e., in polynomial time.
    
    Indeed, if there exist $p\in P$, $t\in [\ell]$ and  $z\in\{1, \dots, |N_{p, t}|\}$
    such that $\textit{sat}_i({\mathbf o})< \alpha(N_{p, t}^z)$ for all 
    $i\in N_{p, t}^z$, then $\mathbf o$ fails to provide EJR, 
    as witnessed by $N_{p, t}^z$.

    Conversely, suppose there is a set $N'$ witnessing that $\mathbf o$
    fails to provide EJR. Let $t$ be the first round such that 
    $\cap_{i\in N'}s_{i, t}\neq \varnothing$, and let $p$ be some
    candidate in $\cap_{i\in N'}s_{i, t}$. Note that 
    $p\in \cap_{i\in N'}s_{i, \tau}$ for all $\tau=t, \dots, n$
    and hence $N'\subseteq N_{p, t}$. Moreover, 
    we have $\beta(N')=\ell-t+1$: the voters in $N'$ agree on $p$
    in rounds $t, \dots, \ell$, and, by the choice of $t$, there
    is no candidate they all agree on in earlier rounds.

    Let $z=|N'|$, and let $u = \max_{i\in N'}\textit{sat}_i({\mathbf o})$;
    since $N'$ witnesses that $\mathbf o$ fails to provide EJR,
    we have $u< \alpha(N')$. 
    Since $\beta(N')=\ell-t+1\le \beta(N_{p, t}^z)$ and $|N'|=z=|N_{p, t}^z|$, 
    we have $\alpha(N')\le \alpha(N_{p, t}^z)$. 
    Also, by Observation~\ref{obs:compare} (applied to the satisfactions 
    of voters in $N'$ and $N_{p, t}^z$),
    it follows that for each $j\in N_{p, t}^z$ we have $\textit{sat}_j({\mathbf o})\le u < \alpha(N')\le \alpha(N_p^z)$, i.e., 
    $N_{p, t}^z$ is also a witness that $\mathbf o$ fails to provide EJR.
    As our algorithm checks $N_{p, t}^z$, 
    it will be able to detect that $\mathbf o$ fails to provide EJR.

To verify whether $\mathbf o$ provides PJR or JR, we modify
the checks that the algorithm performs for each group of voters, 
just as in the proof of Proposition~\ref{prop:m-ell}; 
we omit the details.

\section{Omitted Proofs from Section \ref{sec:s-ejr}}

\subsection{Proof of Theorem \ref{thm:ilp}}
Consider an election $E=(P, N, \ell, ({\mathbf s}_i)_{i\in N})$.

        We first establish that we can assume $|P|\le 2^n$. Indeed, 
        fix a subset of voters $V\subseteq N$ and a round $t\in [\ell]$. If there
        are two candidates $p, p'$ such that for every $i\in N$ we have either $p, p'\in s_{i, t}$
        or $p, p'\not\in s_{i, t}$ then we can remove $p'$ from the approval sets
        of all voters in $N$ at round $t$: every outcome $\mathbf o$ that provides EJR
        and satisfies $o_t=p'$ can be modified by setting $o_t=p$ instead, without affecting
        representation. That is, we can assume that 
        each candidate at round $t$ is uniquely identified by the set of voters who approve her.
        It follows that we can modify the input election $E$ 
        by replacing $P$ with $P'=\{p_V: V\subseteq N\}$ and modifying the voters' approval
        sets accordingly: for each $i\in N$, $t\in [\ell]$ we place $p_V$ in $s'_{i, t}$
        if and only if the approval set $s_{i, t}$ contains a candidate $p$ with 
        $\{i': p\in s_{i', t}\} = V$. The modified election $E'=(P', N, \ell, ({\mathbf s}'_i)_{i\in N})$
        satisfies $|P'|\le 2^n$, and an outcome ${\mathbf o}'$ of $E'$  that provides
        EJR can be transformed into an outcome $\mathbf o$ of $E$ that provides EJR as follows: 
        for each $t\in [\ell]$,
        we set $o_t=p$ for some $p$ such that $o'_t=p_V$ and $p\in \cap_{i\in V}s_{i, t}$.
        Thus, from now on we will assume that $|P|\le 2^n$.
      
        Now, note that each round can be characterized by $n$ voters' approval sets, 
        i.e., a list of length $n$ whose entries are subsets of $P$. Let $\mathbf T$
        denote the set of all such lists; we will refer to elements of $\mathbf T$
        as round {\em types}. There are $2^m$ possibilities for each voter's approval set, 
        and hence $|{\mathbf T}|=(2^m)^n$.
        For each $\tau \in \mathbf T$,
        let $\kappa_\tau \in \mathbb{Z}_{\geq 0}$ 
        be the number of rounds of type $\tau$.
        Overloading notation, we write $p\in s_{i, \tau}$ if $p\in s_{i, t}$
        for a round $t$ of type $\tau$.

       Our ILP will have a variable $x_{p, \tau}$ for each $\tau\in\mathbf T$ and each 
       $p\in P$, 
       i.e., at most $m\cdot (2^m)^n\le 2^n\cdot (2^{2^n})^n$ variables:
       the variable $x_{p, \tau}$ indicates how many times we choose candidate $p$ at a round
       of type $\tau$ and takes values in $0, \dots, \kappa_\tau$. 
       For each $\tau\in\mathbf T$, 
       we introduce a constraint
       \begin{equation}\label{eq:validx}
	\sum_{p\in P}x_{p, \tau}=\kappa_\tau; 
       \end{equation}
       these constraints guarantee
       that solutions to our ILP encode valid outcomes. 

       It remains to introduce constraints
       encoding the EJR axiom. To this end, for each group of voters $V$ we add a constraint
       saying that at least one voter in $V$ has satisfaction at least $\alpha(V)$.

       Observe that the satisfaction of voter $i$ from rounds of type $\tau$ can be written
       as $\sum_{p\in s_{i, \tau}}x_{p, \tau}$; thus, 
       the total satisfaction of $i$ is given by 
       $\sum_{\tau\in\mathbf T}\sum_{p\in s_{i, \tau}} x_{p, \tau}$.
       
       Next, for each $V\subseteq N$ with $V\neq\varnothing$ and each $i\in V$, 
       we introduce a variable $\xi_{i, V}$; these variables take values in $\{0, 1\}$ 
       and indicate
       which voter $i\in V$ receives satisfaction at least $\alpha(V)$. The constraints
       \begin{equation}\label{eq:validxi}
         0\le \xi_{i, V}\le 1, \quad \sum_{i\in V}\xi_{i, V}\ge 1\qquad
                             \text{for all $V\in 2^N\setminus\{\varnothing\}$}
       \end{equation}
       guarantee that for every nonempty subset of voters $V$ we have $\xi_{i, V}=1$
       for at least one voter $i\in V$. Now, we can capture EJR by adding constraints
       \begin{equation}\label{eq:EJR}
         \sum_{\tau\in\mathbf T}\sum_{p\in s_{i, \tau}} x_{p, \tau} \ge\alpha(V)\cdot \xi_{i, V}
         \text{ for all $V\in 2^N\setminus\{\varnothing\}, i\in V$}.        
       \end{equation}
       Indeed, constraint~\eqref{eq:EJR} ensures that the satisfaction of 
       at least one voter in $V$ (one with $\xi_{i, V}=1$)
       is at least $\alpha(V)$.

       We conclude that every feasible solution to the ILP given by 
       the constraints~\eqref{eq:validx}--\eqref{eq:EJR}
       encodes an outcome that provides EJR; moreover, both the number of variables
       and the number of constraints in this ILP can be bounded by functions of $n$.

\subsection{Proof of Proposition \ref{prop:s-EJR_nonexistence}}
Consider an instance with an even number of agents $n=2k$, $k\ge 4$, candidate set $P = \{p_1,\dots,p_n\}$, and $n$ rounds.
    We construct the approval sets as follows. 
    \begin{center}
        \begin{tabular}{ c | c c c c c c c c}
            $s_{i,t}$ & $1$ & $\dots$  & $k$ & $k+1$ & $\dots$ & $n$\\
            \hline
            $1$ & $\{p_1\}$ & $\dots$ & $\{p_1\}$ & $\{p_n\}$ & $\dots$ & $\{p_n\}$\\ 
            $2$ & $\{p_2\}$ & $\dots$ & $\{p_2\}$ & $\{p_n\}$ & $\dots$ & $\{p_n\}$\\  
            $\vdots$ & $\vdots$ & $\vdots$ & $\vdots$ & $\vdots$ & $\vdots$ & $\vdots$\\
            $k$ & $\{p_{k}\}$ & $\dots$ & $\{p_{k}\}$ & $\{p_n\}$ & $\dots$ & $\{p_n\}$\\
            $k+1$ & $\{p_{k+1}\}$ & $\dots$ & $\{p_{k+1}\}$ & $\{p_1\}$ & $\dots$ & $\{p_1\}$\\
            $k+2$ & $\{p_{k+2}\}$ & $\dots$ & $\{p_{k+2}\}$ & $\{p_2\}$ & $\dots$ & $\{p_2\}$\\
            $\vdots$ & $\vdots$ & $\vdots$ & $\vdots$ & $\vdots$ & $\vdots$ & $\vdots$\\
            $n$ & $\{p_n\}$ & $\dots$ & $\{p_n\}$ & $\{p_{k}\}$ & $\dots$ & $\{p_{k}\}$\\  
        \end{tabular}
        \end{center}
        By symmetry, we can assume without loss of generality that in the first $k$ rounds, we select $o_t=p_t$ for $t\in [k]$.
        Note that for each $i=k+1, \dots, n$ the group $N' = \{i\}$ satisfies $\beta(N')=n$ and hence $\alpha(N')=1$.
        Thus, in order to satisfy w-JR, in the next $n/2$ rounds, we must select each of $p_1,\dots, p_{k}$ exactly once.
        However, the resulting outcome $\mathbf o$ fails EJR: indeed,  
        $N''=[k]$ satisfies $\beta(N'')=k$, 
        $\alpha(N'')=k/2\ge 2$, while for each $i\in N''$ we have $\mathit{sat}_i({\mathbf o}) = 1$.
\end{document}